\documentclass[uslettersize, 11 pt, draft, onecolumn]{ieeeconf} 

\IEEEoverridecommandlockouts                              
\usepackage{graphics} 
\usepackage{epsfig} 
\usepackage{mathpazo} 
\usepackage{amsmath} 

\usepackage{amssymb,amsthm}  
\usepackage{amscd}
\usepackage[noadjust]{cite}
\usepackage{float}
\usepackage{times}
\newtheorem{theorem}{Theorem}
\newtheorem{lem}[theorem]{Lemma}
\newtheorem{pro}[theorem]{Proposition}

\newtheorem{Definition}{Definition}
\newtheorem{cor}[theorem]{Corollary}

\newtheorem{Remark}{Remark}
\newcommand{\Sp}{\operatorname{Sp}}
\newcommand{\Co}{\operatorname{Co}}

\newcommand{\bb}{\mathbb}
\renewcommand{\cal}{\mathcal}

\title{\LARGE \bf
Weighted Consensus with Linear Objective Maps}

\author{Xudong Chen, M.-A. Belabbas and Tamer Ba\c sar
\thanks{Xudong Chen, M.-A. Belabbas and Tamer Ba\c sar are with Department of Electrical and Computer Engineering and Coordinated Science Laboratory, University of Illinois at Urbana-Champaign, emails:
        \{xdchen, belabbas, basar1\}@illinois.edu}}

\begin{document}

\maketitle
\thispagestyle{empty}
\pagestyle{empty}

\begin{abstract}  
A consensus system is a linear multi-agent system in which agents communicate to reach a so-called consensus state, defined as the average of the initial states of the agents. Consider a more generalized situation in which each agent is given a positive weight and the consensus state is defined as the weighted average of the initial conditions. We characterize in this paper the weighted averages that can be evaluated in a decentralized way by agents communicating over a directed graph. Specifically,  we introduce  a linear function, called the objective map, that defines the desired final state as a function of the initial states of the agents. We then provide a complete answer to the question of whether there is a decentralized consensus dynamics over a given digraph which converges to the final state specified by an objective map.  In particular, we characterize not only the set of  objective maps that are feasible for a given  digraph, but also  the consensus dynamics that implements the  objective map. In addition, we present a decentralized algorithm to design the consensus dynamics.
\end{abstract}

\section{Introduction}
Consensus algorithms have been recognized  as an important step in a variety of decentralized and distributed algorithms, such as the rendezvous problem, distributed convex optimization or distributed sensing.
We pose and solve in this paper what we term the \emph{weighted consensus} problem over a  directed graph. Specifically, given a set of positive weights assigned to the agents, we say that the agents reach a weighted consensus if they converge to the weighted average of their initial conditions -- a formal definition to be given shortly.   As is commonly done, we assume that the information flow in the system is described by a directed graph. Our goal is to determine which weighted averages can be computed for a given information flow.  Furthermore, we  describe how the agents can communicate over the graph to \emph{design} the dynamical {system} whose evolution reaches the desired weighted consensus.  Computing a weighted average rather than a simple average is a natural one when the agents in the system are not all on equal footing. For example, think of a rendezvous problem where the rendezvous position depends on the  initial positions of only a small group of agents; of distributed sensing, where the weighting can be proportional to the accuracy of the sensing device, or of opinion dynamics, where participants may have different levels of influence on the decision process. 
Because of their broad relevance, a fair amount is already known about consensus algorithms. Indeed, questions concerning sufficient and/or necessary conditions for agents to reach consensus (\cite{JT1,AJ,RM1,LM,RW,SM,JT2,LJ3,CM}), questions  
concerning time delay (\cite{RM1},~\cite{CM}), consensus with quantized measurements (\cite{TB1},~\cite{TB2}), consensus with time varying network topologies (\cite{JT1,AJ,RM1,LM,RW,SM,JT2,LJ3,CM}), and questions about  convergence rate (\cite{LJ1,LJ2,TB2,TB5}), robustness (\cite{KBT},\cite{ET}) in the presence of an adversary  have all been treated to some degree.

Broadly speaking, the problem we address in this paper is one of \emph{feasibility of an objective given decentralization constraints}. Similar questions, but involving controllability of linear systems~\cite{lin_structural},  stability of linear systems~\cite{belabbas_sparse} and formation control~\cite{lorenzen_belabbas_ecc2014} have also been investigated. While the general problem of feasibility of an objective under decentralization constraints is far from being completely understood, we shall see that a fairly complete characterization can be obtained in the present case, in the cases of both discrete- and continuous-time dynamics. However, still open questions remain, such as: How to handle negative weights? 
How to handle time-varying information flow graphs? How to make sure that no-agent can ``game the system" and increase or decrease its assigned weight?

We next describe the model  precisely. We assume that there are $n$ agents $ x_1,\ldots,  x_n$ evolving in $\mathbb{R}^d$, and that the underlying network topology is specified by  a directed graph (or simply {\it digraph}) $G = (V,E)$, with $V = \{1,\ldots, n\}$ the set of vertices and $E$ the set of edges. We let $V^-_i$ be a subset of $V$ comprised of the outgoing neighbors of vertex $i$, i.e., 
\begin{equation*}
V^-_i := \{j\in V\mid i\to j\in E\}
\end{equation*}
and we assume in this paper that each agent $ x_i$ can only observe its outgoing neighbors.  The equations of motion for the $n$ agents $ x_1,\ldots, x_n$ are then given by
\begin{equation}\label{MODEL}
\frac{d}{dt}{ x}_i = \sum_{j\in V^-_i}a_{ij} \cdot ( x_j -  x_i), \hspace{10pt} \forall i=1,\ldots,n 
\end{equation}
with each $a_{ij}$ a non-negative real number, which we call the interaction weight.

The objective of the system is characterized by positive real numbers $w_i$. We define the \emph{objective function} $f: \mathbb{R}^{n\times d}\to \mathbb{R}^d$ as:
\begin{equation*}
f( x_1,\ldots, x_n) := \sum^n_{i=1}w_i x_i.
\end{equation*}  

The feasibility question we ask is the following: given a digraph $G=(V,E)$, and  a weight vector $ w = (w_1,\ldots, w_n)$ in $\mathbb{R}^n$, does there exist a set of non-negative interaction weights $\{a_{ij}\mid i\to j\in E\}$ such that for any initial condition $ x_1(0),\ldots, x_n(0)$ in $\mathbb{R}^d$, all agents will converge to the same point in $\mathbb{R}^d$ specified by the objective map, i.e.,  
\begin{equation*}
\lim_{t\to \infty} x_i(t) =   f\big ( x_1(0),\ldots,  x_n(0)\big )
\end{equation*}
for all $i=1,\ldots, n$. In other words, we require that  all the agents not only reach consensus, but also converge to a specific point which is a weighted sum of the initial positions of the agents. 
In the following section, we will convert this problem to one  of asking whether there exists a sparse, infinitesimal stochastic matrix $A$ with a fixed zero pattern (specified by the digraph) such that  $A$ has a simple zero eigenvalue with the vector $ w$ being the corresponding left eigenvector.

 
In the paper, we will provide a complete answer to the question of weighted consensus within model~\eqref{MODEL}. In particular, we will characterize not only the set of  objective maps which are feasible by choices of interaction weights, but also the set of interaction weights for a feasible objective map. Note that the problem of evaluating averages in a \emph{distributed} manner has also been handled using discrete-time dynamics~\cite{DLI1,DLI2,DLI3,HT}. 

The paper is organized as follows. In section~\ref{DefThm}, we introduce some  definitions and  state the main results of the paper. Precisely, the main theorem characterizes the set of objective maps that can be realized over a given network topology.  Section III is devoted to the proof of the main theorem. In Section VI, we present a decentralized algorithm for finding a set of interaction weights associated with a feasible objective map. In particular, we relate the set of interaction weights to solutions of graph balancing. We provide conclusions in the last section.  The paper ends with an Appendix.

\section{Definitions, Problem Reformulation and \\The Main Theorem}\label{DefThm}
In this section, we  introduce the main definitions used in this work, formulate the  weighted consensus problem in precise terms,  and state the main result of the paper.

\subsection{Background and Notation}

We consider in this paper only  \emph{simple}  directed graphs, that is directed graphs with no self loops, and with at most one edge between each ordered pair of vertices. We denote by $G = (V,E)$  a directed graph where $V$ is the vertex set and $E$ the edge set. We denote by $i\to j$ an edge of $G$, with $i$ and $j$ the start-vertex and the end-vertex of the edge, respectively.  A vertex $r$ is said to be a {\bf root} of $G$ if for all $i \in V$, there is a path from $i$ to $r$ . We say that  $G$ is {\bf rooted} if it contains a root.   Graphs with only one vertex are by convention rooted. We denote by $V_r \subset V $ the set of all roots of $G$. The digraph $G$ is {\it strongly connected} if, for any ordered pair of vertices $(i,j)$, there is a path from $i$ to $j$. In this case, all vertices of $G$ are roots, i.e.,  $V_r = V$.     It is well known that if the  digraph $G$ associated with system~\eqref{MODEL} is rooted,  then all agents  converge to the same state  for all initial conditions (see, for example, \cite{RW}). Conversely, if for any initial condition, all agents of system~\eqref{MODEL} converge to the same state, then the underlying digraph must be rooted. Hence, we only consider rooted digraphs as the underlying digraphs of system~\eqref{MODEL}. For a subset $V' \subset V$, we call $G'$ a subgraph of $G$ \emph{induced} by $V'$ if $G' = (V',E')$ and $E'$ contains any edge of  $E$ whose start-vertex and end-vertex are in $V'$. We have the following definition:

\begin{Definition}[Relevant Subsets]\label{def:relsub}
Let $G = (V,E)$ be a rooted digraph, and $V'$ be a subset of $V$. We say $V'$ is {\bf relevant to $G$} (or simply {\bf relevant})   if it satisfies the two conditions:    

\begin{enumerate}
\item[a).] The set $V'$ is contained in the root set $V_r$;

\item[b).] The subgraph $G'$ induced by $V'$ is strongly connected. 
\end{enumerate}
\end{Definition}


For $G$ a digraph with $n$ vertices, we can always let $V=\{1,2,\ldots,n\}$. We denote by $\operatorname{Sp}[V]$  the $(n-1)$-simplex contained in $\mathbb{R}^n$ with vertices the standard basis vectors $e_1, \ldots, e_n \in \mathbb{R}^n$.  For $V' \subset V$, we define similarly $\operatorname{Sp}[V']$ as the convex hull of $\{e_i \mid i \in V'\}$:
$$\operatorname{Sp}[V'] := \left\{ \sum_{i \in V'} \alpha_i e_i \mid \alpha_i \geq 0, \sum_{i \in V'} \alpha_i = 1 \right\}.$$
We use the notation $\operatorname{Sp}(V')$ to denote the \emph{interior} of $\operatorname{Sp}[V']$:
$$\operatorname{Sp}(V') := \left\{ \sum_{i \in V'} \alpha_i e_i \mid \alpha_i > 0, \sum_{i \in V'} \alpha_i = 1 \right\}.$$
If $V'$ is comprised of only  one vertex, say vertex $i$,  we then set  $\operatorname{Sp}[V']=\operatorname{Sp}(V')  = \{ e_i\}$. We introduce a similar notation to denote a convex cone. Let $C_i$, for $i=1, \ldots, l$, be vectors in $\bb R^m$; we denote the convex cone spanned by $C_i$ by
$$\Co[C_1,\ldots,C_l] :=  \left\{ \sum_{i=1}^l \alpha_i C_i \mid \alpha_i \geq 0 \right\}.$$ We denote its interior by $$\Co(C_1,\ldots,C_l) :=  \left\{ \sum_{i=1}^l \alpha_i C_i \mid \alpha_i > 0 \right\}.$$


\begin{Definition}[Sparse Infinitesimal Stochastic Matrix]
We say  a matrix $A$ is an {\bf infinitesimal stochastic matrix (ISM)} if its off-diagonal entries are non-negative, and its rows sum to zero. 
Let $G$ be a digraph with $n$ vertices. We let $\mathbb{A}_G$ be the set of $n$-by-$n$ ISMs with the following properties:  $A=(a_{ij})\in \mathbb{A}_{G}$ if, for $i\neq j$, we have 
\begin{equation*}
a_{ij} =\left\{
\begin{array}{ll}
\ge 0 & \text{if } i\to j \text{ is an edge of }G\\
=0 & \text{otherwise}
\end{array}\right.
\end{equation*}
If $G$ consists of only one vertex, then $\mathbb{A}_G = \{0\}$ is a singleton.
\end{Definition}

Let $\mathbf{1}$ be a vector of all ones in $\mathbb{R}^n$; then for each matrix $A$ in $\mathbb{A}_{G}$, we have $A\mathbf{1} = 0$. So each matrix $A$ has at least one zero eigenvalue. Furthermore, it is well known that if $A$ is an ISM, then the real parts of eigenvalues of $A$ are non-positive. In particular, if the digraph $G$ is rooted and  $a_{ij}>0$ for each  $i\to j \in E$, then the matrix $A$ has a simple zero eigenvalue. 




\subsection{Main Results}
We start by formulating the targeted consensus problem in view of the facts introduced above. First, note that we can rewrite~\eqref{MODEL} into a matrix form as follows: Let $X$ be an $n$-by-$d$ matrix with $ x^\top_i$ the $i$-th row of $X$. Then, system~\eqref{MODEL} is equivalent to
\begin{equation}\label{MODELN}
\dot X = A X
\end{equation}  
with matrix $A$ contained in $\mathbb{A}_G$.  For the purpose of reaching consensus, we require that the matrix $A$ have a simple zero eigenvalue. Let $ w\in \operatorname{Sp}[V]$ be the left eigenvector of $A$ corresponding to the zero eigenvalue. Then, for any initial condition $X(0)$, we have 
\begin{equation*}
\lim_{t\to \infty}X(t) =  \mathbf{1} \cdot w^\top X(0)
\end{equation*}
If we write $ w = (w_1,\ldots, w_n)$, then
\begin{equation*}
\lim_{t\to \infty} x_i(t) = \sum^n_{i=1} w_i  x_i(0) 
\end{equation*}
Conversely, if the expression above holds for all initial conditions, then the matrix $A$ must have $w$ as a left eigenvector and zero as a simple eigenvalue. We thus introduce the following subset of $\mathbb{A}_G$:

\begin{Definition}[$w$-Feasible Dynamics]
Let $ w\in \operatorname{Sp}[V]$.  We define the set of {\bf $w$-feasible dynamics} $\mathbb{A}_{G}(w) \subset \mathbb{A}_G$ as the set of ISMs satisfying:  

\begin{enumerate}
\item The matrix $A$ has a simple zero eigenvalue. 
\item The vector $ w$ is the left eigenvector of $A$ corresponding to the zero eigenvalue, i.e., $A^\top w = 0$.
\end{enumerate}

\end{Definition}


So the question we raised in the first section can be restated as follows: For a given digraph $G$ and a vector $w \in \operatorname{Sp}[V]$, is the set $\mathbb{A}_{G}(w)$  empty?   We answer this question in Theorem~\ref{MAIN}:

\begin{theorem}\label{MAIN}
Let $G = (V,E)$ be a rooted digraph, and $V_1,\ldots, V_q$ be the relevant subsets  of $G$. Let $W$ be a subset of $\operatorname{Sp}[V]$ comprised of vectors $ w$ for which $\mathbb{A}_{G}(w)$ is  nonempty. Then, $$W = \cup^q_{i=1}\operatorname{Sp} (V_i).$$
\end{theorem}

We  now consider briefly the case of  discrete-time dynamics and show that essentially the same result holds. To be precise, we consider the model    
\begin{equation*}
X(k+1) = \bar A X(k)
\end{equation*}
where $\bar A$ is a stochastic matrix (i.e. its entries are all non-negative and the sum of the entries in any given row is one.) 
Of course, the matrix $\bar A$ is also a sparse matrix, with the zero pattern specified by  the rooted graph $\bar  G = (V,E)$. For simplicity,  we  assume that  $\bar G$ has self-loops, i.e., $i\to i\in E$ for all $i\in V$. Thus, each diagonal entry of $A$ is allowed to take a positive value.  Similarly, we let $\bar{\mathbb{A}}_{\bar G}$ be the collection of the sparse stochastic matrices associated with $\bar G$, and let $\bar{\mathbb{A}}_{\bar G}(w)\subset \bar{\mathbb{A}}_{\bar G}$ be  such that each $A\in \bar{\bb{A}}_{\bar G}(w)$ has a simple eigenvalue $1$ with $\bar A^{\top} w =  w$ (note that if $\lambda$ is another eigenvalue of $\bar A$, then $|\lambda|<1$ by the Perron-Frobenius theorem).

\begin{theorem}\label{main2}
Let $\bar{G} = (V,E)$ be a rooted digraph with a self-loop at each vertex. Let $V_1,\ldots, V_q$ be the relevant subsets  of $\bar G$. Let $\overline{W}$ be a subset of $\operatorname{Sp}[V]$ comprised of vectors $ w$ for which $\bar{\mathbb{A}}_{\bar{G}}(w)$ is  nonempty. Then, $$\overline{W} = \cup^q_{i=1}\operatorname{Sp} (V_i).$$

\end{theorem}
\begin{proof}
Let $G$ be the digraph $\bar G$ with self-arcs removed, and let $W$ be the set of $w$ for which $\bb{A}_G(w)$ is non-empty. We prove Theorem~\ref{main2} by showing that $\overline{W}=W$ and then appealing to Theorem~\ref{MAIN}.  We first  pick $w \in W$ and let $A\in \mathbb{A}_{G}(w)$.  For sufficiently small $\epsilon>0$, we have $(I + \epsilon A)\in \bar{\mathbb{A}}_{\bar G}(w)$. Thus, $W \subseteq \overline{W}$. Conversely, let $w \in \overline{W}$ and $\bar A\in \bar{\mathbb{A}}_{\bar G}(w)$, then $(\bar A - I)\in {\mathbb{A}}_{{G}}(w)$. Thus, $\overline{W}\subseteq W$. This completes the proof.
\end{proof}

The next section is devoted to the proof of Theorem~\ref{MAIN} and is organized as follows. In section~\ref{SPMS}, we  focus on the relevant subsets of $V$. In particular, we show that if the set $\mathbb{A}_{G}(w)$ is nonempty, then the vector $ w$ has to be  in the union of $\operatorname{Sp}(V_1),\ldots,\operatorname{Sp}(V_q)$. In section~\ref{3b}, we investigate $w$-feasible dynamics without the requirements that $G$ be rooted and that $A$ has zero as a simple eigenvalue. With this relaxation, we prove that all the relaxed $w$-feasible dynamics form a closed \emph{convex cone}. 
In section~\ref{3c}, we assume that $G$ is strongly connected, and introduce the notion of principal subset to characterize $\mathbb{A}_G(w)$.  In particular, we provide a canonical decomposition of $\mathbb{A}_G(w)$ into disjoint  subsets. Moreover, we show that  the closure of $\bb{A}_G(w)$ is  the closed convex cone of the relaxed $w$-feasible dynamics. We then combine these results and prove Theorem~\ref{MAIN} in section~\ref{3d}.


\section{Relevant Subsets of Vertices, Cycles of Digraphs and Principal Subsets of Cycles}

\subsection{On Relevant Subsets}\label{SPMS}

We derive here some preliminary relations between the set of $w$-feasible dynamics $\mathbb{A}_G(w)$ and the relevant subsets of $G$ introduced in Definition~\ref{def:relsub}. To this end, we set

\begin{equation}\label{eq:defVw}
V_w := \{ i \in V \mid w_i \neq 0 \}.
\end{equation} 
We establish the following result:

\begin{pro}\label{PMS}
Let $G = (V,E)$ be a rooted digraph, and $ w$ be a vector in $\operatorname{Sp}[V]$. If $\mathbb{A}_{G}(w)$ is nonempty, then $V_{w}$ is a  relevant subset of $G$. 
\end{pro}

The proof of Proposition~\ref{PMS} proceeds by  first showing that the subset $V_w$ is contained in the root set $V_r$ of $G$, and then  showing that the subgraph $G_{ w}$, derived by restricting $G$ to $V_{ w}$, is strongly connected. This is done in Lemmas~\ref{2lem1} and~\ref{2lem2} below.

\begin{lem}\label{2lem1}
Let $G$ be a rooted digraph, and $ w$ be a vector in $\operatorname{Sp}[V]$. If $\mathbb{A}_{G}(w)$ is nonempty, then $V_{ w}$ is a subset of $V_r$. 
\end{lem}

\begin{proof}
Without loss of generality, we may assume that the root set $V_r$ consists of the first $m$ vertices. Then, each matrix $A$ in $\mathbb{A}_{G}$ is a lower block-triangular matrix, i.e., 
\begin{equation}\label{A11}
A = 
\begin{pmatrix}
A_{11} & 0 \\
A_{21} &  A_{22}
\end{pmatrix}
\end{equation}
with $A_{11}$ an $m$-by-$m$ square matrix and $A_{12}$ a zero block. Indeed, if $a_{ij} \neq 0$ for $1 \leq i\leq m$, then $i \to  j$ is an edge of $G$, and since $i$ is a root, then so is $j$.

In view of the above,  the exponential $\exp(At)$, as the state transition matrix of system~\eqref{MODELN}, is also a lower block-triangular matrix with blocks of the same dimensions as the blocks of $A$. Furthermore, since the matrix $A$ has a simple zero eigenvalue while all of its other eigenvalues have negative real parts,  we have 
\begin{equation}\label{ew}
\lim_{t\to \infty }\exp(At) = \mathbf{1} \cdot  w^\top
\end{equation}
Using~\eqref{A11} and~\eqref{ew}, we know that $w_i = 0$ for all $i=m+1,\ldots, n$. This proves the result. 
\end{proof}

We now show that the subgraph $G_{ w}$ is strongly connected.

\begin{lem}\label{2lem2}
Let $G$ be a rooted digraph, and $ w$ be a vector in $\operatorname{Sp}[V]$. If $\mathbb{A}_{G}(w)$ is nonempty, then the subgraph $G_{ w}$ is strongly connected. 
\end{lem}

\begin{proof}
As in the proof of the previous lemma, we may assume without loss of generality that the set $V_{ w}$ consists of the first $m$ vertices of $G$. Let $A$ be a matrix in $\mathbb{A}_{G}(w)$, and partition $A$ into blocks as
\begin{equation*}
A = 
\begin{pmatrix}
A_{11} & A_{12}\\
A_{21} & A_{22}
\end{pmatrix}
\end{equation*} 
with $A_{11}$ being an $m$-by-$m$ matrix and correspondingly, partition $ w$ into 
\begin{equation*}
 w = ( w',0)
\end{equation*}
with $ w'$ a vector in $\mathbb{R}^{m}$. By assumption,  each entry of $ w'$ is nonzero.  Since $A$ is in $\mathbb{A}_{G}(w)$, we have $A^\top  w = 0$, so then $ A^\top_{12} w' = 0$. Because each entry of $A_{12}$ is non-negative and  each entry of $ w'$ is positive, we must have $A_{12} = 0$. This then implies that $A_{11}$ is an $m$-by-$m$ ISM.

Let $G' = (V',E')$, with $V':= \{1,\ldots, m\}$, be a subgraph of $G$ induced by the block matrix $A_{11}$, i.e.,  an edge $i\to j$ is in $E'$ if and only if $a_{ij} >0$.  It suffices to show that $G'$ is strongly connected. To do this, note that the digraph $G'$ must be rooted because otherwise $A_{11}$, and hence $A$, has at least two zero eigenvalues.  Now, suppose that $G'$ is not strongly connected,  then the root set of $G'$, denoted by $V'_r$, is a proper subset of $V'$. On the other hand, if we let $V'_{ w'}$ be the collection of indices of nonzero entries of $ w'$, then $V'_{ w'} = V'$. But, from  Lemma~\ref{2lem1}, we know that
\begin{equation*}
V' = V'_{ w'} \subseteq V'_r \subsetneq V' 
\end{equation*}
which is a contradiction. Thus, we conclude that $G'$ is strongly connected. This completes the proof.  
\end{proof}

By combining Lemmas~\ref{2lem1} and~\ref{2lem2}, we establish Proposition~\ref{PMS}.  
We conclude this section by describing the relevant sets of some families of digraphs, namely cycle graphs and complete graphs.

\begin{cor}
Let $G$ be an $n$-cycle. If $\mathbb{A}_{G}(w)$ is nonempty, then either $ w =  e_i$ for some $i=1,\ldots,n$ or $ w$ has no zero entry.  
\end{cor}

\begin{proof}
If $G$ is a cycle, then $G$ is strongly connected. On the other hand, there is no proper subgraph of $G$, other than a single vertex, which is strongly connected. 
\end{proof}

\begin{cor}
Let $G$ be a rooted digraph, and let $W$ be the set of vectors in $\operatorname{Sp}[V]$ with $\mathbb{A}_{G}( w)$ being nonempty. If $W = \operatorname{Sp}[V]$, then the digraph $G$ must be a complete graph. 
\end{cor}

\begin{proof}
It suffices to show that each induced subgraph of two vertices is a $2$-cycle. Let 
\begin{equation*}
 w  = \frac{1}{2}( e_i +  e_j).
\end{equation*}
Then, the set of nonzero entries of $ w$ is given by
\begin{equation*}
V_{ w} = \{i,j\}
\end{equation*}  
So by Proposition~\ref{PMS}, the set $V_{ w}$ is relevant. In particular, the subgraph $G_{ w}$ is strongly connected, and hence a $2$-cycle.  
\end{proof}

\subsection{On Cycles of Digraphs}\label{3b}
In this sub-section, we assume that $G = (V,E)$ is an arbitrary digraph. Note that if $G$ is not rooted and $w \in \operatorname{Sp}[V]$, then $\mathbb{A}_G(w)$ is empty since no matrix in $\mathbb{A}_G$ has zero as a simple eigenvalue. We thus relax this condition in the following definition:

\begin{Definition}[Relaxed $w$-Feasible Dynamics] Let $ w\in \operatorname{Sp}[V]$.  We define the set of {\bf relaxed $w$-feasible dynamics} $\widehat{\mathbb{A}}_{G}(w) \subset \mathbb{A}_G$ as follows:  
\begin{equation*}\label{hatA}
\widehat{\mathbb{A}}_{G}(w):= \{A \in \mathbb{A}_{G}\mid A^\top w = 0\}
\end{equation*} 
\end{Definition}

Our goal in this sub-section is to characterize the set $\widehat{\mathbb{A}}_{G}(w)$. This is important because as we will see later  when $G$ is strongly connected, the set $\widehat{\mathbb{A}}_{G}(w)$ is the closure of $\mathbb{A}_{G}(w)$. 

We say that a digraph $G'$ is a  {cycle of $G$} if $G'$ is  a subgraph of $G$ and is a cycle with at least two vertices.  We label the  cycles of $G$ as  $G_1,\ldots, G_k$.  Let $ w$ be in $\operatorname{Sp}(V)$. So then,  each entry $w_i$ of $ w$ is positive. For each cycle $G_i $ of $G$, we now define an {\bf associated  ISM} $C_i$ by specifying its off-diagonal entries. Let $C_{i, jk}$ be the $jk$-th entry of $C_i$, and let
\begin{equation}\label{eq:defci}
C_{i,jk} := \left\{
\begin{array}{ll}
1/w_j  &  \text{if } j\to k \text{ is an edge of } G_i\\
0 & \text{otherwise}
\end{array}
\right.
\end{equation}
Its diagonal entries are set so that the entries of each row of $C_i$ sum to zero. We establish the following result in this sub-section.

\begin{pro}\label{CH}
Let $G$ be a digraph, and $ w$ be a vector in $\operatorname{Sp}(V)$.  Let $G_1,\ldots, G_k$ be the cycles of $G$, and  $C_1,\ldots, C_k$ be the associated ISMs. If $k \geq 1$, then the set $\widehat{\mathbb{A}}_{G}(w)$ is a convex cone spanned by $C_1,\ldots, C_k$, i.e., 
\begin{equation*}
\widehat{\mathbb{A}}_{G}(w)= \Co[C_1,\cdots, C_k]
\end{equation*} 
Moreover, each ray $\{\alpha C_i \mid \alpha \ge 0\}$ is an extreme ray of the convex cone.  
\end{pro}
By convention, if $k=0$, we set $\widehat{\mathbb{A}}_{G}(w) =\{0\}$. We note here that a similar result relating cycles and doubly stochastic matrices can be found in  \cite{CycDsm}. We prove Proposition~\ref{CH} by first investigating a special case where $G$ is acyclic, i.e.,  there is no cycle contained in $G$.

\begin{lem}\label{3lem1} 
Let $G$ be an acyclic digraph, and $ w$ be a vector in $\Sp(V)$.  Then, $\widehat{\mathbb{A}}_{G}(w) = \{0\}$.
\end{lem}

\begin{proof}
We prove the lemma by induction on the number $n$ of vertices of $G$.  
\\
{\it Base case}. If $G$ consists of only one vertex, then there is nothing to prove.
\\
{\it Inductive step}. Suppose the statement of the lemma holds for $n$, we show that it holds for $n+1$.

Since $G$ is acyclic, there must exist a vertex, say vertex $1$, with no incoming edge. Let $A\in \widehat{\mathbb{A}}_{G}(w)$, and let $ a_1$ be the  first column of $A$. Then $ a_1$ has at most one nonzero entry, i.e., the first entry of $ a_1$. Let $a_{11}$ be the first entry of $ a_1$, then 
\begin{equation*}
  a_1^{\top} w = w_1 a_{11} = 0.
\end{equation*}       
Since $w_1$ is positive by assumption, we then have $a_{11} = 0$. This then implies that the first row vector of $A$ is also a zero vector. Hence we can write $A$ as
\begin{equation*}
A = 
\begin{pmatrix}
0 & 0\\
0 & A'
\end{pmatrix} 
\end{equation*}
with $A'$ an $n$-by-$n$ matrix and $V'=\{2,\ldots,n+1\}$ the corresponding vertex set. It now suffices to show that $A'$ is a zero matrix.  

Let $ w'$ be a vector in $\mathbb{R}^{n}$ defined by
\begin{equation*}
 w':= \frac{1}{1- w_1} (w_2,\ldots,w_{n+1}).
\end{equation*} 
Note that $ w'$ is well defined since $w_1 < 1$. By construction, we have $ w'\in \Sp(V')$. 
Moreover, $ w'$ satisfies the condition that $A'^{\top}  w'= 0$. 
Let $G'$ be a subgraph of $G$ induced by $V'$. Let $\mathbb{A}_{G'}$ be the set of sparse  ISMs associated with $G'$, and let 
\begin{equation*}
\widehat{\mathbb{A}}_{G'}(w'):= \left \{A'\in \mathbb{A}_{G'}\mid  A'^\top  w' = 0 \right \}
\end{equation*} 
Then, $A'\in \widehat{\mathbb{A}}_{ G'}(w')$. Since the subgraph $G'$ is acyclic,  we conclude by the induction hypothesis that $\widehat{\mathbb{A}}_{ G'}(w')$ contains only the zero matrix, and hence $A' = 0$. This  completes the proof.  
\end{proof}

We now prove Proposition~\ref{CH}.

\begin{proof}[Proof of Proposition~\ref{CH}] 

We first show that $ \Co[C_1,\cdots, C_k]\subseteq \widehat{\mathbb{A}}_{G}(w)$. It suffices to show that each $C_i$ is contained in $\widehat{\mathbb{A}}_{G}(w)$. 
Denote by $ v_j$ the $j$-th column of $C_i$; then either $ v_j$ is a zero vector or $ v_j$ contains two nonzero entries.  If $ v_j$ is a zero vector, then $  v_j^\top  w= 0$, so we focus on the latter case. By definition of $C_i$, the $j$-th entry of $ v_j$ is $-1/w_j$. We assume without loss of generality that the other nonzero entry of $ v_j$ is the $k$-th entry; its value is then given by $1/w_k$.  Then, 
\begin{equation*}
  v_j^\top  w = w_j \cdot (-1/w_j) + w_k \cdot (1/w_k) = 0.
\end{equation*}
This equality holds for each column vector of $C_i$, and hence  $ C_i^{\top}  w = 0$.

We now show that $\widehat{\mathbb{A}}_{G}(w) \subseteq \Co[C_1,\cdots, C_k]$. To this end,  fix a matrix $A$ in $\widehat{\mathbb{A}}_{G}(w)$. Assume that there is a cycle $G_i = (V_i,E_i)$ in $G$  such that $a_{jk}>0$ for all $j\to k \in E_i$. We show below that if this is not the case, $A$ is necessarily the zero matrix. First, note that there exists an $\alpha>0$ such that $(A - \alpha C_i)$ is also an ISM; indeed the matrix $A - \alpha C_i$ is  in $\widehat{\mathbb{A}}_{G}(w)$ because 
\begin{equation*}
(A - \alpha C_i)^\top  w= A^\top  w -\alpha C_i^\top  w = 0.
\end{equation*}
Now let
\begin{equation*}
\alpha_i := \max \left \{\alpha \in\bb{R}\mid (A - \alpha C_i) \in \widehat{\mathbb{A}}_{G}(w) \right \}, 
\end{equation*}  
which can be computed explicitly as
\begin{equation*}
\alpha_i = \min\left \{ w_{j} a_{jk} \mid j\in V_i, j\to k\in E_i\right\}.  
\end{equation*}
Let $A' : = A - \alpha_i C_i$, then $A'$ has more zero entries than does $A$. To see this,  it suffices to check the off-diagonal entries of $A'$. First note that if the $jk$-th, $j\neq k$, entry of $A'$ is positive, then so is the $jk$-th entry of $A$. Thus, $A'$ has at least as many zero entries as $A$ does. 
Now let
\begin{equation*} 
j\to k \in \operatorname{argmin} \{ w_{j} a_{jk} \mid j\in V_i, j\to k\in E_i\}. 
\end{equation*}
Then, the $jk$-th entry of $A'$ is zero because
\begin{equation*}
a'_{jk}=a_{jk} - (w_{j}a_{jk}) C_{i,jk} = a_{jk} - (w_{j}a_{jk}) 1/w_{j} = 0.
\end{equation*}
On the other hand, we have $a_{jk} > 0$. Thus, $A'$ has more zero entries than does $A$. We then say that the matrix $A'$ is a {\it reduction} of $A$. 
Now let 
\begin{equation*}
A\to A^{(1)} \to A^{(2)}\to\ldots  
\end{equation*} 
be a chain of reductions. Since $A^{(k)}$ has more zero entries than $A^{(k-1)}$ does, the chain must be finite. Suppose that this chain stops at $\tilde A$, i.e., there does not exist a reduction of $\tilde A$. It then suffices to prove that $\tilde A$ is a zero matrix. 
Let $\tilde G$ be a digraph with $n$ vertices induced by the matrix $\tilde A$. 
Since there is no reduction of  $\tilde A$, the induced digraph $\tilde G$ must be acyclic. Since $\tilde A^{\top}  w = 0$ with $ w\in \operatorname{Sp}(V)$, by Lemma~\ref{3lem1}, we have $\tilde A = 0$.

It now remains to show that each ray $\{\alpha C_i \mid \alpha\ge  0\}$ is an extreme ray of the convex cone $\widehat{\mathbb{A}}_{G}(w)$. We show that for each matrix $C_i$, there does not exist a set of non-negative coefficients $\alpha_j$'s such that  
\begin{equation*}\label{extray}
C_i = \sum_{j\neq i}\alpha_j C_j. 
\end{equation*}
We prove by contradiction. Suppose that the expression above holds; then at least one coefficient $\alpha_j$,  for $j\neq i$, is positive. Hence, the cycle $G_j$ is a subgraph of $G_i$. On the other hand, $G_i$ itself is a cycle, so we must have $G_i = G_j$, which is a contradiction. This  completes the proof. 
\end{proof}

\subsection{On Principal Subsets}\label{3c}

In this section, we assume that  $G=(V,E)$ is a strongly connected digraph with $n$ vertices for $n>1$. Let $\mathcal{G}=\{G_1,\ldots,G_k\}$ be the set of cycles of $G$;   $\mathcal{G}$ is  non-empty   since for each edge $i\to j$ of $G$, there is at least one cycle containing that edge. 
Recall that a Hamiltonian cycle of $G$ is a cycle that passes through every vertex. We can thus say that this cycle  strongly connects the graph, in the sense that each vertex is connected to every other vertex using only edges in the cycle. If $G$ does not have a Hamiltonian cycle, one can nevertheless use several cycles to strongly connect $G$. We make this formal through the introduction of  \emph{principal subset} of $\mathcal{G}$. In words, these are sets of cycles of $G$ that strongly connect $G$. Specifically, we have the following definition:
\begin{Definition}
[Principal Subset] Let $G$ be a digraph and $\mathcal{G}=\{G_1,\ldots,G_k\}$ be its set of cycles. Let $G_i=(V_i,E_i)$. We call a subset $\{G'_{1},\ldots,G'_{m}\}$ of $\mathcal{G}$ {\bf principal}
if the graph
\begin{equation*}
G'=(V,\cup_{i} E'_i)
\end{equation*}
is strongly connected. We label the principal subsets of $\mathcal{G}$ as  $\mathcal{G}_1,\ldots,\mathcal{G}_p$. 
\end{Definition}




Let  $\mathcal{G}_i=\{G_{i_1},\ldots,G_{i_m}\}$ be a principal subset of $\mathcal{G}$. Fix a vector $ w\in \operatorname{Sp}(V)$,  and let $C_{i_1},\ldots,C_{i_m}$ be the associated ISMs (defined in~\eqref{eq:defci}). Recall that 
$$\Co[C_{i_1},\ldots,C_{i_m} ]= \left \{\sum^m_{j=1}\alpha_{j}C_{i_j}\mid \alpha_{j}\ge 0\right \} $$
is the convex cone spanned by $C_{i_1},\ldots,C_{i_m}$, and $\Co(C_{i_1},\ldots,C_{i_m})$ is the interior of $\Co[C_{i_1},\ldots,C_{i_m}]$. 
With a slight abuse of notation, we write $$\Co[\mathcal{G}_i]:=\Co[C_{i_1},\ldots,C_{i_m}],$$ and $$\Co(\cal{G}_i) :=  \Co(C_{i_1},\ldots,C_{i_m}).$$   
Equipped with definitions and notations above, we now prove the following result:

\begin{pro}\label{4lem1}
Let $G$ be a strongly connected digraph with $n$ vertices for $n> 1$, and let $\mathcal{G}_1,\ldots, \mathcal{G}_p$ be principal subsets of $\mathcal{G}$. Let $ w\in \operatorname{Sp}(V)$. Then, $$\mathbb{A}_{G}(w) = \cup^p_{i=1}\Co(\cal{G}_i).$$  
\end{pro}

\begin{proof}
We first show that each set $\Co(\cal{G}_i)$ is  contained in $\mathbb{A}_{G}(w)$. Suppose that $\mathcal{G}_i$ is comprised of cycles $G_{i_1},\ldots, G_{i_m}$. For any matrix  $A$ in $\Co(\cal{G}_i)$,  there exists a set of positive coefficients $\alpha_{1},\ldots, \alpha_{m}$ such that 
\begin{equation*}
A = \sum^m_{j = 1}\alpha_{j} C_{i_j}
\end{equation*}
Let $G_A$ be the digraph induced by matrix $A$; then $G_A$ is strongly connected by definition of $\cal{G}_i$. Consequently, the matrix $A$ has a simple zero eigenvalue. Furthermore, we have 
\begin{equation*}
C^{\top}_{i_j}  w = 0
\end{equation*} 
for all $j=1,\ldots, m$. So then, $A^{\top}  w = 0$, and hence $A\in \mathbb{A}_{G}(w)$. We have thus proved that  $\Co(\cal{G}_i)$ is  contained in $\mathbb{A}_{G}(w)$.

Next we show that the set $\mathbb{A}_{G}(w)$ is contained in the union of $\Co(\cal{G}_{1}),\ldots, \Co(\cal{G}_{p})$.  Let $A$ be a matrix in $\mathbb{A}_{G}(w)$; then $A$ is also contained in $\widehat{\mathbb{A}}_{G}(w)$. Thus by Proposition~\ref{CH}, there is a set of non-negative coefficients $\alpha_1,\ldots,\alpha_k$ such that 
\begin{equation}\label{eq:Adefconc}
A = \sum^k_{i=1} \alpha_i C_i 
\end{equation} 
Suppose $\alpha_{i_1},\ldots, \alpha_{i_m}$ are the non-zero coefficients out of $\alpha_1,\ldots, \alpha_k$, we then let $\mathcal{G}':=\{ G_{i_1},\ldots,G_{i_m}\}$. We need to show that $\mathcal{G}'$ is a principal subset of $\mathcal{G}$. 
Let $G_A$ be the digraph induced by $A$.  It follows from~\eqref{eq:Adefconc} and~\eqref{eq:defci} that $G_A$ is the union of $G_{i_1},\ldots,G_{i_m}$. It now suffices to show that $G_A$ is strongly connected. Suppose that it is not the case; then by Proposition~\ref{PMS}, the set $\mathbb{A}_{G}(w)$ is empty because $V_w = V$ which is not a relevant set of $G_A$. On the other hand, $A\in \bb{A}_G(w)$, which is a contradiction. Hence, $G_A$ is strongly connected, and thus $\mathcal{G}'$ is a principal subset of $\mathcal{G}$. 
\end{proof}

We note that in general, the right-hand-side of the decomposition of $\bb{A}_G(w)$  given in Proposition~\ref{4lem1} may contain redundant terms.  Indeed, we might have for some $i$ that   
\begin{equation*}
\Co(\cal{G}_i) \subseteq \cup_{j\neq i}\Co(\cal{G}_j)
\end{equation*} 
To rule out this redundancy, we introduce the following definition:

\begin{Definition}[Minimal Cover]
For a collection of arbitrary sets  $\{\cal{A}_i\}^p_{i=1}$,  let $\cal{A} = \cup_{i=1}^p \cal{A}_i$. We say that $\cal{A}_{i_1},\ldots,\cal{A}_{i_l}$ is a {\bf minimal cover} of   $\cal{A}$ with respect to the collection $\{\cal{A}_i\}^p_{i=1}$  if $$\cal{A} = \cup_{j=1}^l \cal{A}_{i_j}$$ and $l$ is the least integer for the relation above to hold.

\end{Definition}

For the collection of sets $\Co(\cal{G}_1),\ldots,\Co(\cal{G}_p)$, it should be clear that the minimal cover of $\bb{A}_G(w)$ always exists. 
We now show that, quite surprisingly, there exists a {\it unique} minimal cover of $\bb{A}_G(w)$ and its components are \emph{pairwise disjoint}. To do so, we define a partial order over the set of principal subsets as follows: let us introduce the notation 
\begin{equation*}
\label{eq:defpartialorder}
\cal{G}_i \succ \cal{G}_j \mbox{ if } \Co(\cal{G}_i)\supsetneq \Co(\cal{G}_j).\end{equation*} A principal subset $\cal{G}'$ is said to be a {\bf maximal element} with respect to the partial order if there does not exist another principal subset $\cal{G}''$ such that $\cal{G}''\succ \cal{G}'$. 
We label  $\cal{G}^*_{1}, \ldots, \cal{G}^*_{l}$ as the maximal elements.
With the definitions above, we establish the following result. 

\begin{pro}\label{mrg}
The sets  $\Co(\cal{G}^*_{1}),\ldots, \Co(\cal{G}^*_{l})$ are pairwise disjoint. Moreover,
\begin{equation}\label{eq:mindecompAG}
\mathbb{A}_G(w) = \Co(\cal{G}^*_{1}) \cup\ldots \cup \Co(\cal{G}^*_{l})  
\end{equation}
is the unique minimal cover of $\mathbb{A}_G(w)$.  
\end{pro}

To prove Proposition~\ref{mrg}, we first establish the following fact:

\begin{lem}\label{parlem0}
If $\cal{G}_i \succ \cal{G}_j$, then $\cal{G}_i \supsetneq \cal{G}_j$. 
\end{lem}

\begin{proof}
It suffices to show that $\cal{G}_i \supseteq \cal{G}_j$. Suppose it is not the case, and that there exists a cycle $G_k$ such that 
$
G_k \in \cal{G}_j - \cal{G}_i
$. 
By Proposition~\ref{CH},  $\{\alpha C_k\mid \alpha \ge 0\}$ is an extreme ray of  $\widehat{\bb{A}}_G(w)$. This, in particular, implies that  $C_k\notin \Co[\cal{G}_i]$. On the other hand, if $\cal{G}_i\succ \cal{G}_j$, then $\Co[\cal{G}_i]\supseteq \Co[\cal{G}_j]$. But then,
$$
C_{k}\in \Co[\cal{G}_i] - \Co[\cal{G}_j] = \emptyset
$$
which is a contradiction. This completes the proof.  
\end{proof}

The converse of Lemma~\ref{parlem0} is not true in general, i.e., we may not conclude  that $\cal{G}_i \succ \cal{G}_j$ from the condition that $\cal{G}_i \supsetneq \cal{G}_j$. However, we will be able to prove that if $\cal{G}_i\supsetneq \cal{G}_j$, then either $\cal{G}_i\succ\cal{G}_j$  or $\Co(\cal{G}_i) \cap \Co(\cal{G}_j) = \emptyset$.

For a given principal subset $\cal{G}' = \{G_{i_1},\ldots, G_{i_m}\}$, we denote by 
$$L(\cal{G}') := \left\{\sum_{j=1}^m \beta_j C_{i_j} \mid \beta_j \in \bb{R}\right\}$$ 
the linear space spanned by the matrices $\{C_{i_1},\ldots, C_{i_m}\}$.  Note that each $\Co(\cal{G}_i) $ is an open subset of  $L(\cal{G}_i)$. It should be clear that if $\cal{G}_i \supseteq \cal{G}_j$, then $L(\cal{G}_i)  \supseteq L(\cal{G}_j)$.  Furthermore, we have the following fact:


\begin{lem}\label{parlem1} 
Assume that  $\cal{G}_i \supsetneq \cal{G}_j$, and let $\cal{G}' := \cal{G}_i - \cal{G}_j$. Then, the following three properties hold:
\begin{itemize}
\item[1.] If $L(\cal{G}_i) = L(\cal{G}_j)$, then  $\cal{G}_i \succ \cal{G}_j$.
\item[2.] If $L(\cal{G}_i) \supsetneq L(\cal{G}_j)$ and $\Co(\cal{G}') \cap  L(\cal{G}_j) \neq \emptyset$, then $\cal{G}_i \succ \cal{G}_j$.
\item[3.] If $L(\cal{G}_i) \supsetneq L(\cal{G}_j)$ and $\Co(\cal{G}') \cap L(\cal{G}_j) = \emptyset$,  then $\Co(\cal{G}_i) \cap \Co(\cal{G}_j) = \emptyset$.
\end{itemize}
\end{lem}

We refer readers to the Appendix for a proof of Lemma~\ref{parlem1}. We are now in a position to prove Proposition~\ref{mrg}. 

\begin{proof}[Proof of Proposition~\ref{mrg}]
First we show that Relation~\eqref{eq:mindecompAG} holds.
Since the $\cal{G}^*_i$'s are maximal elements,  each set $\Co(\cal{G}_j)$, for $j = 1,\ldots,p$, is contained in some $\Co(\cal{G}^*_i)$ for $i=1,\ldots,l$. By Proposition~\ref{4lem1}, we have that
\begin{equation*}
\mathbb{A}_G(w) = \Co(\cal{G}_1)\cup\ldots\cup \Co(\cal{G}_p),
\end{equation*} from which~\eqref{eq:mindecompAG} follows.

Next, we will show that the    sets $\Co(\cal{G}^*_1),\ldots, \Co(\cal{G}^*_l)$ are pairwise disjoint. We prove this by contradiction. Assume that there is a matrix
\begin{equation*}
A  \in \Co(\cal{G}^*_i)\cap\Co(\cal{G}^*_j), \hspace{10pt} \mbox{for }i\neq j. 
\end{equation*} 
Set $\cal{G}':= \cal{G}^*_i \cup \cal{G}^*_j$. It is not hard to see that $A\in \Co(\cal{G}')$. Indeed, if we let $\cal{G}^*_i = \{G_{i_1},\ldots, G_{i_m}\}$ and $\cal{G}^*_j = \{G_{j_1},\ldots,G_{j_{m'}}\}$, then there are positive coefficients $\alpha_{i_t}$'s and $\beta_{j_s}$'s such that 
\begin{equation*}
A = \sum^m_{t = 1}\alpha_{i_t} C_{i_t} = \sum^{m'}_{s = 1}\beta_{j_s} C_{j_s}.
\end{equation*}
Using the previous equation,  we can express $A$ as
\begin{equation*}
A = \frac{1}{2} \left( \sum^m_{t = 1}\alpha_{i_t} C_{i_t} +  \sum^{m'}_{s = 1}\beta_{j_s} C_{j_s}\right ).
\end{equation*} which shows that $A \in \Co(\cal{G}')$. 
Since the $\cal{G}_k^*$'s are distinct, we have $\cal{G}' \supsetneq \cal{G}^*_i $ and/or $\cal{G}' \supsetneq \cal{G}^*_j $. We assume without loss of generality that the former holds. By Lemma~\ref{parlem1}, either $\cal{G}'\succ \cal{G}^*_i$ or $\Co(\cal{G}') \cap \Co(\cal{G}^*_i) = \emptyset$. But since 
$$
A\in \Co(\cal{G}') \cap \Co(\cal{G}^*_i),
$$
we conclude that $\cal{G}'\succ \cal{G}^*_i$ which then contradicts the fact that $\cal{G}^*_i$ is a maximal element. Thus, we have proved that $\Co(\cal{G}^*_1),\ldots, \Co(\cal{G}^*_l)$ are pairwise disjoint.

It now remains to show that~\eqref{eq:mindecompAG} is a minimal cover, and indeed it is the unique minimal cover. Let $\{\cal{G}'_1,\ldots, \cal{G}'_{l'}\}$ be any set of principal subsets such that
\begin{equation}\label{fakemr}
\mathbb{A}_G(w) =  \Co(\cal{G}'_1)\cup \ldots \cup  \Co(\cal{G}'_{l'})
\end{equation} 
is a minimal cover of $\mathbb{A}_G(w)$. Then, we  have $l' \leq l$.

Now to each $\cal{G}_i'$, we can assign a maximal element in $\overline{\cal{G}'_{i}} \in \{\cal{G}^*_1 , \ldots,\cal{G}^*_{l}\}$ such that $\overline{\cal{G}'_{i}} \succ\cal{G}'_{i}$. We claim that  any such assignment has to satisfy the condition that  if $i \neq j$, then $\overline{\cal{G}'_{i}} \neq \overline{\cal{G}'_{j}}$. To see this,  note that  for any pair $\cal{G}'_i,\cal{G}'_j$, there is no  principal subset $\cal{G}'$ such that  $\cal{G}' \succ \cal{G}'_i$ and $\cal{G}' \succ \cal{G}'_j$. Indeed, if it were the case, then $ \Co(\cal{G}'_i) \cup\Co(\cal{G}'_j) \subseteq \Co(\cal{G}') $ and we can replace $ \Co(\cal{G}'_i) \cup\Co(\cal{G}'_j)$ with $ \Co(\cal{G}')$, which contradicts the assumption that~\eqref{fakemr} is a  minimal cover of $\bb{A}_G(w)$. 

Since $\overline{\cal{G}'_{i}} \neq \overline{\cal{G}'_{j}}$ for all $i\neq j$, we can assume, without loss of generality, that $\overline{\cal{G}'_{i}} =\cal{G}^*_i$  for all $i=1,\ldots, l'$.  It should be clear that 
$$ \Co(\cal{G}'_i) \subseteq  \Co(\cal{G}^*_i) $$ 
and by Lemma~\ref{parlem0}, the equality holds if and only if $\cal{G}'_i = \cal{G}^*_i$. Since $ \Co(\cal{G}^*_1),\ldots, \Co(\cal{G}^*_l)$ are pairwise disjoint; we thus conclude that 
\begin{equation*}
 \Co(\cal{G}'_1)\cup \ldots \cup   \Co(\cal{G}'_{l'})\subseteq \Co(\cal{G}^*_1)\cup \ldots \cup   \Co(\cal{G}^*_l)
\end{equation*}
and the equality holds if and only if 
$$\{\cal{G}'_1,\ldots, \cal{G}'_{l'}\} = \{\cal{G}^*_1,\ldots, \cal{G}^*_{l}\}.$$  
In other words, we have shown that~\eqref{eq:mindecompAG} is the unique minimal cover of $\mathbb{A}_G(w)$.
\end{proof}

We conclude this sub-section by relating $\mathbb{A}_G(w)$ to $\widehat{\bb{A}}_{G}(w)$:

\begin{pro}\label{SC}
Let $G$ be a strongly connected digraph, and $ w$ be a vector in $\operatorname{Sp}(V)$.  Then, $\mathbb{A}_{G}(w)$ is a nonempty convex set and its closure is $\widehat{\mathbb{A}}_{G}(w)$. 
\end{pro} 

\begin{proof}
 If  $G$ consists of only one vertex, then $\mathbb{A}_{G}(w) = \widehat{\mathbb{A}}_{G}(w) = \{0\}$. Henceforth,  we assume that the number of vertices of $G$ is greater than one.

We first show that $\mathbb{A}_{G}(w)$ is a convex set. 
Let $A_i$ and $A_j$ be two matrices in $\Co(\cal{G}_i)$ and $\Co(\cal{G}_j)$ respectively; we show that for $\alpha_i$ and $\alpha_j$  positive,  the matrix $\alpha_i A_i + \alpha_j A_j$ is contained in $\mathbb{A}_{G}(w)$.  First note that the matrix  $\alpha_i A_i + \alpha_j A_j$ is  an element in $\Co(\cal{G}_i\cup\cal{G}_j) $.  Since $\cal{G}_i\cup\cal{G}_j$ is a principal subset, $\alpha_i A_i + \alpha_j A_j$ is contained in $\mathbb{A}_{G}(w)$.

It now remains to show that the closure of $\mathbb{A}_{G}(w)$ is $\widehat{\mathbb{A}}_{G}(w)$. First we note that $\mathbb{A}_{G}(w)$ is contained in $\widehat{\mathbb{A}}_{G}(w)$ while $\widehat{\mathbb{A}}_{G}(w)$ is a closed set. So the closure of $\mathbb{A}_{G}(w)$ must be contained in $\widehat{\mathbb{A}}_{G}(w)$. We now show that the converse is also true, that is $\widehat{\mathbb{A}}_{G}(w)$ is contained in the closure of $\mathbb{A}_{G}(w)$. Choose a matrix $A$ in $\widehat{\mathbb{A}}_{G}(w) - \mathbb{A}_{G}(w)$; then by Proposition~\ref{CH}, we have
\begin{equation*}
A = \sum^k_{i=1}\alpha_i C_i 
\end{equation*}
with each $\alpha_i$ non-negative. Since for each positive real number $\epsilon>0$, the matrix 
\begin{equation*}
A(\epsilon) :=  \sum^k_{i=1}(\alpha_i + \epsilon) C_i 
\end{equation*}
is contained in $\mathbb{A}_{G}(w)$, $A$ is in the closure of $\mathbb{A}_{G}(w)$. This  completes the proof.
\end{proof}

\subsection{Proof of Theorem~\ref{MAIN}}\label{3d}

We now return to the proof of Theorem~\ref{MAIN} stated in Section~\ref{DefThm}. 

\begin{proof}[Proof of Theorem~\ref{MAIN}]
Recall that $V_1,\ldots, V_q$ are the relevant subsets of $G$, and $\Sp(V_i)\subset \operatorname{Sp}[V]$ is the interior of the convex hull  spanned by  the unit vectors $\{ e_{j}\mid j \in V_i\}$. Also recall that $W$ is the set of vectors $ w$ in $\operatorname{Sp}[V]$ for which $\mathbb{A}_{G}(w)$ is  not empty.

From Proposition~\ref{PMS}, we know that if $\bb{A}_G(w)$ is not empty, then $V_w$ (defined in~\eqref{eq:defVw}) is relevant and thus, the set $W$ is contained in the union  $\cup^q_{i=1} \Sp(V_i)$. We  now show that the converse is also true. Let $V'$ be a relevant subset of $V$, and for simplicity,  assume that $V' = \{1,\ldots, m\}$ with $m\le n$. Let $G'$ be the subgraph of $G$ induced by  $V'$; then by the definition of relevant subset, $V'$ is contained in the root set of $G$ and $G'$ is strongly connected. 
Let $ w$ be a vector contained in $\Sp(V')$,  and let $ w'$ be the vector in $\mathbb{R}^m$ containing the first $m$ entries of $ w$, i.e., $ w  = ( w',0)$. 

We   now prove that $\mathbb{A}_{G}(w)$ is nonempty  by  constructing a matrix $A$ in it. We partition the matrix $A$ into four blocks as 
\begin{equation}\label{eq:pttA}
A = 
\begin{pmatrix}
A_{11} & A_{12}\\
A_{21} & A_{22}
\end{pmatrix}
\end{equation}
with $A_{11}$  an $m$-by-$m$ matrix. Let $\bb{A}_{G'}(w')$ be the set of $m$-by-$m$ ISMs associated with the digraph $G'$ and the vector $ w'$. Then by Proposition~\ref{SC}, the set $\bb{A}_{G'}(w')$ is nonempty, and hence we can pick  $A_{11}$  in $\bb{A}_{G'}(w')$. Let  $A_{12}$ be the zero matrix. 
Choose $A_{21}$ and $A_{22}$ such that if $i\to j$ is an edge of $G$ and if $i> m$, then the $ij$-th entry of $A$ is positive.  We now show that the resulting matrix $A$ is contained in $\mathbb{A}_{G}(w)$. First note that by the choice of $A_{11}$, we have
$
 A^{\top}_{11}  w' = 0
$, and hence $A^\top w = 0$.

It now suffices to show that $A$ has zero as a simple eigenvalue.  
Let $G_{A}$ be the digraph induced by the matrix $A$.  We will show that $G_A$ is rooted with  root set  $V'$. Since $A_{12} = 0$ by construction,  there is no edge $i\to j$ with $1\le i\le m$ and $j> m$.  Thus, for any vertex $j\notin V'$,  there is no path in $G_A$ from a vertex $i\in V'$ to $j$. So then, the root set of $G_A$ is a subset of $V'$. 
On the other hand, each vertex in $V'$ is a root of $G$, and by construction of $A_{21}$ and $A_{22}$, we know that if  $i\to j$, for  $i> m$,  is an edge of $G$, then it is also an edge  of  $G_A$. Thus, for any vertex $i\notin V'$, there is a path from $i$ to some vertex in $V'$. Since the subgraph $G'$ is strongly connected, the set $V'$ is the root set of $G_A$. Hence, $A$ has zero as a simple eigenvalue. This  completes the proof.
\end{proof}

\begin{cor}
Let $G$ be a rooted graph, and $w$ be such that $\bb{A}_G(w)$ is not empty. Then, 
\begin{itemize}
\item[1.] The digraph $G_A$ induced by any matrix $A\in \bb{A}_G(w)$ is rooted, and $V_w$ is the root set. \item[2.] The set $\bb{A}_G(w)$ is convex.
\end{itemize}
\end{cor}

\begin{proof}
Letting $A\in \bb{A}_G(w)$, we show that $G_A$ is rooted. For simplicity, but without loss of generality, we assume that $w = (w',0)$ with $w'\in\bb{R}^m$ containing the nonzero entries of $w$. Partition $A$ into blocks as we did in~\eqref{eq:pttA}. Then,  $A_{12} = 0$ because 
$A_{12}^\top w' =0$. Thus, the root set of $G_{A}$ is a subset of $V_w$.    
On the other hand, from Lemma~\ref{2lem2}, we know that the subgraph $G_w$ induced by  $V_w$ is strongly connected. Thus, the root set of $G_{A}$ is $V_w$.

Now let $A'$ and $A''$ be in $\mathbb{A}_{G}(w)$, and let 
\begin{equation*}
A := \alpha' A' + \alpha'' A''
\end{equation*} 
with $\alpha'$ and $\alpha''$ positive. We show that $A\in \bb{A}_G(w)$. It should be clear that $A^\top  w = 0$. It remains to show that $A$ has zero as a simple eigenvalue. This holds because the two induced digraphs  $G_{A'}$ and $G_{A''}$ are rooted. Hence, the digraph $G_A$, as a union of $G_{A'}$ and $G_{A''}$, is also rooted.  Thus, we conclude that $A$ has a simple zero eigenvalue with $A^{\top}  w =0$ and hence,  $A\in \mathbb{A}_G(w)$. This completes the proof.
\end{proof}

\section{Decentralized Implementation of $w$-Feasible Dynamics}
In this section, we will assume that $G $ is a rooted digraph, and $ w$ is a vector for which the set $\mathbb{A}_{G}(w)$ is nonempty. We present here a decentralized algorithm that allows the agents to find a matrix $A$ in $\mathbb{A}_{G}(w)$. In particular, we assume  that each agent $ x_i$ only knows its own weight $w_i$, and each agent is only able to communicate/cooperate with its neighbors, which are defined as the agents connected to $i$ with either an incoming or an outgoing edge. 

The implementation of the algorithm derived here relies on decentralized methods for the so-called \emph{graph balancing problem} for the digraph $G$ (referred to as $G$-balancing). We say that the coefficients  $b_{ij} \geq 0$, for $i\to j\in E$, form a solution of $G$-balancing if for each vertex $i$ of $G$, we have
\begin{equation*}
\sum_{k\in V^+_i}b_{ki} = \sum_{j\in V^-_i} b_{ij} 
\end{equation*} 
with $V^+_i$ and $V^-_i$ being respectively the incoming and the outgoing neighbors of vertex~$i$. We call a solution positive (resp. non-negative) if the $b_{ij}$'s are strictly positive (resp. non-negative).

\begin{lem}\label{7lem1}
Let $G$ be a digraph, and let $ w$ be a vector in $\operatorname{Sp}(V)$. Let $\mathbb{B}_{\ge 0}$ be the set of non-negative solutions of the $G$-balancing problem. Let $\Lambda$ be a diagonal matrix with $\Lambda_{ii}=w_i$. Then  $\mathbb{B}_{\ge 0}=\Lambda^{-1}\widehat{\mathbb{A}}_{G}(w)$.   
\end{lem}

\begin{proof}
Let $\{b_{ij}\mid i\to j\in E \}$ be a non-negative solution of $G$-balancing. This solution gives rise to an  ISM $B$, i.e., if we let $B_{ij}$ be the $ij$-th, $i\neq j$, entry of $B$, then 
\begin{equation*}
B_{ij} := \left\{
\begin{array}{ll}
b_{ij} & \text{if } i\to j \in E\\
0 & \text{otherwise}.
\end{array}
\right. 
\end{equation*} 
Since $\{b_{ij}\mid i\to j\in E\}$ is a non-negative solution of the  $G$-balancing problem, not only is $B$ an ISM, but so is $B^\top$. In other words, we have $B ^\top \mathbf{1} = 0$. Let $\Lambda$ be a diagonal matrix with $ w$ being its diagonal; then the matrix $A:=\Lambda^{-1}B$ is an element in $\widehat{\mathbb{A}}_{G}(w)$ because 
\begin{equation*}
A^\top  w = B^\top\Lambda^{-1}_{ w} w = B^\top \mathbf{1} = 0
\end{equation*}  
Conversely, if $A$ is a matrix in $\widehat{\mathbb{A}}_{G}(w)$, then the matrix $B := \Lambda_{ w} A$ yields a non-negative solution of $G$-balancing. Moreover, this map between $\mathbb{B}_{\ge 0}$ and $\widehat{\mathbb{A}}_{G}(w)$ is one-to-one and onto because the diagonal matrix $\Lambda$ is invertible. 
\end{proof}
\begin{Remark} 
It is known that  $G$ is strongly connected if and only if  there exists a positive solution of $G$-balancing. We also note that for $G$  strongly connected and $\{b_{ij}\mid i\to j\in E\}$  a positive solution of $G$-balancing,  the matrix $A = \Lambda^{-1} B$ is contained in $\mathbb{A}_{G}(w)$.  
\end{Remark}

Now suppose that $G$ is strongly connected, and  $ w$ is in $\operatorname{Sp}(V)$.   Suppose that there is a decentralized algorithm  to find a positive solution $\{b_{ij}\mid i\to j\in E\}$ of $G$-balancing. Then by Lemma~\ref{7lem1}, it suffices for each agent $ x_i$ to set the interaction weights as
\begin{equation*}
a_{ij}: = b_{ij}/w_i,
\end{equation*}
because the resulting set $\{a_{ij}\mid i\to j\in E\}$ yields a matrix $A$ in $\mathbb{A}_{G}(w)$. So for $G$  a strongly connected digraph and $ w$  a vector in  $\operatorname{Sp}(V)$, the problem of finding a matrix  $A\in \mathbb{A}_{G}(w)$ is reduced to the problem of finding a positive solution of $G$-balancing. This is a well-studied problem, and we provide here a decentralized iterative algorithm for agents to find a positive solution of  $G$-balancing

{\bf Algorithm A1: $G$-balancing for $G$ strongly connected}. We let $b_{ij}[l]$ be the value of $b_{ij}$ at iteration step $l\ge 0$. We assume that at every step, the agent $ x_i$ knows the values of $b_{ki}$ for all $k\in V^+_i$ and the value of $b_{ij}$ for all $j\in V^-_i$. 
\\
{\it Initialization}. Each agent $ x_i$ sets $b_{ij}[0] = 1$ for all $j\in V^-_i$. 
\\
{\it Iterative step}. Each agent $ x_i$ updates $b_{ij}[l]$ as 
\begin{equation*}
b_{ij}[l+1] = \frac{1}{2} \left (b_{ij}[l] +  \sum_{k\in V^+_i} b_{ki}[l] / |V^-_i| \right )
\end{equation*} 
We  refer  to \cite{DecAlg} for a proof of convergence of the algorithm. 

We  now    consider the case of $G$  rooted, but not necessarily strongly connected. We assume that  the vector $ w$ is chosen so that $\mathbb{A}_{G}(w)$ is nonempty.   For simplicity, we still assume that only the first $m$ entries of $ w$ are nonzero, and let $ w'$ be in $\mathbb{R}^m$ so that $ w = ( w',0)$.
Let $G'$ be the subgraph of $G$ induced by the first $m$ vertices, and let
\begin{equation*}
\bb{A}_{G'}(w'):= \left \{A' \in \mathbb{A}' \mid A'^\top  w' = 0 \right \}
\end{equation*}
Similarly, we partition an ISM $A$ into four blocks  
\begin{equation*}
A = 
\begin{pmatrix}
A_{11} & A_{12}\\
A_{21} & A_{22}
\end{pmatrix}
\end{equation*}
with $A_{11}$ being $m$-by-$m$. We will now describe a decentralized algorithm for agents to construct the  four block matrices $A_{11},A_{12}, A_{21}$ and $A_{22}$ so that the resulting matrix  $A$ is contained in $\mathbb{A}_{G}(w)$. The algorithm, described below, implements the construction of a specific choice of  matrix $A \in \bb{A}_G(w)$ given in the proof of Theorem~\ref{MAIN}; to be precise,  we want the blocks $A_{ij}$'s to be such that
\begin{itemize}
\item[1.] The block matrix $A_{11}$ is contained in $\mathbb{A}_{G'}(w')$. 
\item[2.] The block matrix $A_{12}$ is a zero matrix.
\item[3.] The two block matrices $A_{21}$ and $A_{22}$ satisfy the condition that if $i\to j$ is an edge of $G$ with $i>m$, then $a_{ij}$, the $ij$-th entry of $A$, is $1$.    
\end{itemize}

{\bf Algorithm A2: Computing $A\in \mathbb{A}_{G}(w)$ for $G$ rooted}. We use an iterative method to derive $a_{ij}$, but  since the graph $G$  is not necessarily strongly connected,  we need to process the input as  described below.
\\
{\it Initialization}.  Each agent $ x_i$ informs his/her neighbors (both incoming and outgoing) of his/her  own weight $w_i$, and agent $ x_i$ receives the information of weights  from all of his/her neighbors. 
There are two different cases depending on whether the weight $w_i$ of agent $ x_i$ is zero or not. 
\begin{itemize}
\item[a).] If $w_i = 0$, then agent $ x_i$ sets $a_{ij}[0] = 1$ for all $j\in V^-_i$.
\item[b).] If $w_i > 0$, then agent $ x_i$ defines the sets 
\begin{equation*}
\left\{
\begin{array}{l}
V'^+_{i} := \left \{k\in V^+_i  \mid w_k > 0  \right \} \\ 
V'^-_{i} := \left \{j\in V^-_i  \mid w_j > 0  \right \}
\end{array}
\right.
\end{equation*} 
We note that $V'^+_i$ (resp. $V'^-_i$) is just the set of incoming (resp. outgoing) neighbors of $i$ in the digraph $G'$. Agent $ x_i$ then sets
\begin{equation*}
a_{ij}[0] = \left\{
\begin{array}{ll}
1/w_i  & \text{if } j\in V'^-_i  \\
0 & \text{otherwise}
\end{array}
\right.
\end{equation*} 
\end{itemize}  
{\it Iterative step}. We still consider two cases:
\begin{itemize}
\item[a).] If $w_i = 0$, then agent $ x_i$ retains the value of $a_{ij}[l]$, i.e., $a_{ij}[l+1] =a_{ij}[l] = 1$ for all $j \in V^-_i$. 
\item[b).] If $w_i>0$, then agent $ x_i$ updates $a_{ij}[l]$ as 
\begin{equation*}
a_{ij}[l+1] = \left\{
\begin{array}{ll}
 \frac{1}{2} \left (a_{ij}[l] + \frac{\sum_{k\in V'^+_i} w_k\cdot a_{ki}[l]}{w_i \cdot |V'^-_i|} \right )  & \text{if } j\in V'^-_i  \\
0 & \text{otherwise}
\end{array}
\right.
\end{equation*} 
In other words, agent $ x_i$ only updates $a_{ij}[l]$ with $j\in V'^-_i$. 
\end{itemize} 

Note that if we replace $a_{ij}[\cdot] $ with $b_{ij}[\cdot]/w_i$, then we actually recover Algorithm A1 and obtain a positive solution of $G'$-balancing.

\section{Conclusions}
In this paper, we have worked with  the standard  consensus model, and addressed the question of given a rooted digraph $G$, what kind of linear objective map
\begin{equation*}
f( x_1,\ldots,  x_n) = \sum^n_{i=1} w_i x_i
\end{equation*}
with $ w = (w_1,\ldots,w_n)\in \operatorname{Sp}[V]$, is feasible by a choice of interaction weights $a_{ij}$? By introducing the notion of relevant subsets of vertices, we have provided a complete answer to this question in Theorem~\ref{MAIN} for the case of continuous-time dynamics, and in Theorem~\ref{main2} for the case of discrete-time dynamics. We   illustrated the results on the particular case of circles and complete graphs. In addition, we have also dealt with implementation of  a feasible objective map $f$.  By looking at cycles of $G$, and introducing the notion of principal subsets, we have shown how the set of feasible objective maps $f$ can be related to a decomposition of the space of stochastic matrices in Propositions~\ref{4lem1},~\ref{mrg}, and~\ref{SC}. Finally,  we also presented  a decentralized algorithm for agents in a network to implement a selected set of interaction weights that achieves a feasible objective map. 

Future work may focus on the case where the interaction weights $a_{ij}$'s are allowed to be negative. Note that in the case when $a_{ij}$'s are non-negative, the vector $ w$ associated with a feasible objective map has to be in the unit simplex. Thus, if $f$ is an objective map with $ w\notin \operatorname{Sp}[V]$, and if there is a choice of $a_{ij}$ under which  $f$ is feasible, then there must exist some $a_{ij}$ which is negative. The question about feasibility, and the question about decentralized implementation can still be  raised in this context for a given digraph $G$. Other open problems, such as dealing with  time-varying digraphs, dealing with nonlinear objective maps, and dealing with the presence of a malicious player who attempts to increase his/her own weight, as in \cite{ET},  are all interesting topics to look at.

\section*{Acknowledgement}
The authors acknowledge useful discussions with Ji Liu at CSL, UIUC.

\section*{Appendix}
We prove here Lemma~\ref{parlem1} stated in section~\ref{3c}. 

\begin{proof}[Proof of Lemma~\ref{parlem1}]
We prove Lemma~\ref{parlem1} case by case. 

\subsubsection{Proof for the case where $L(\cal{G}_i) = L(\cal{G}_j)$}
We first show that $\Co(\cal{G}_i)\supseteq \Co(\cal{G}_j)$, then show that  $\Co(\cal{G}_i) \neq \Co(\cal{G}_j)$.
Assume that $\cal{G}_i  = \{G_{1},\ldots,G_{m}\} $ and $\cal{G}_j = \{G_{1},\ldots, G_{m'}\}$ with $m> m'$. 
Since $L(\cal{G}_i) = L(\cal{G}_j)$, for each $G_s$ with $s> m'$,  there are coefficients $\sigma_{st}$'s such that
\begin{equation}\label{eq:coeffst}
C_s = \sum^{m'}_{t=1} \sigma_{st} C_t 
\end{equation}
Now choose a matrix
\begin{equation*}
A = \sum^{m'}_{t=1} \alpha_t C_t \in \Co(\cal{G}_j)
\end{equation*}
with each $\alpha_t>0$, and construct positive coefficients $\beta_{s}$'s such that
\begin{equation*}
A = \sum^{m}_{s=1} \beta_s C_s \in \Co(\cal{G}_i)
\end{equation*}
Choose an $\epsilon>0$, and define
\begin{equation*}
\beta_s :=
\left\{
\begin{array}{ll}
\alpha_s - \epsilon \sum^m_{t=m'+1} \sigma_{ts} & \mbox{if } 1\le s\le m'\\
\epsilon & \mbox{if } s > m'.
\end{array}
\right.
\end{equation*}
Since $\alpha_s>0$, we can choose $\epsilon$ sufficiently small so that $\beta_s>0$ for all $s=1,\ldots,m$. With this choice of $\beta_s$'s, we have
\begin{equation*}
\sum^{m}_{s=1} \beta_s C_s  =  \sum^{m'}_{s=1} (\alpha_s - \epsilon \sum^m_{t=m'+1} \sigma_{ts}) C_s + \epsilon \sum^m_{s = m'+1} C_s 
\end{equation*}
After rearranging the terms, we obtain 
\begin{equation*}
\sum^{m}_{s=1} \beta_s C_s = \sum^{m'}_{s=1} \alpha_s C_s  + \epsilon \sum^m_{s = m'+1}(C_s - \sum^{m'}_{t=1} \sigma_{st} C_t )
\end{equation*}
By~\eqref{eq:coeffst}, we have 
$$
A = \sum^{m'}_{s=1} \alpha_s C_s = \sum^{m}_{s=1} \beta_s C_s\in  \Co(\cal{G}_i).
$$

Next we show that there exists a matrix $A \in \Co(\cal{G}_i) - \Co(\cal{G}_j)$. Consider the cycle $G_m\in \cal{G}_i - \cal{G}_j$. By Proposition~\ref{CH},  $\{\alpha C_m \mid\alpha>0\}$ is an extreme ray of $\widehat{\bb{A}}_{G}(w)$, and hence,  there does not exist a set of non-negative coefficients $\alpha_i$'s such that  
$
C_m = \sum^{m'}_{i=1} \alpha_i C_i
$. 
This, in particular, implies that
\begin{equation*}
\inf\left \{ \| C_m - C\| \mid C\in \Co[\cal{G}_j] \right\} >0
\end{equation*}
where $\|C_m - C\|$ is the trace of $(C_m - C)^{\top}(C_m - C)$.  In other words, the matrix $C_m$ and the convex cone $\Co[\cal{G}_j]$ are separable. 
Thus, if we choose
\begin{equation*}
A := C_m + \epsilon \sum^{m-1}_{i=1}C_i \in  \Co(\cal{G}_i)
\end{equation*}
for sufficiently small $\epsilon>0$, then $A\notin \Co[\cal{G}_j]$. This completes the proof.


\subsubsection{Proof for the case where $L(\cal{G}_i) \supsetneq L(\cal{G}_j)$ and $\Co(\cal{G}') \cap \Co(\cal{G}_j) \neq \emptyset$} 
We first show that $\Co(\cal{G}_i) \supseteq \Co(\cal{G}_j)$, and then show that $\Co(\cal{G}_i) \neq \Co(\cal{G}_j)$. Let $\cal{G}_i  = \{G_{1},\ldots,G_{m}\} $ and $\cal{G}_j = \{G_{1},\ldots, G_{m'}\}$ with $m> m'$. Choose a matrix 
\begin{equation*}
A = \sum^{m'}_{t=1} \alpha_t C_t \in \Co(\cal{G}_j)
\end{equation*}
and construct positive coefficients $\beta_{s}$'s such that
\begin{equation*}
A = \sum^{m}_{s=1} \beta_s C_s \in \Co(\cal{G}_i)
\end{equation*}
Since $\Co(\cal{G}')$ intersects $L(\cal{G}_j)$, there are positive coefficients $\tilde \beta_s$'s and coefficients $\gamma_t$'s such that 
\begin{equation}\label{eq:coeff2}
\sum^m_{s = m'+1} \tilde\beta_s C_s = \sum^{m'}_{t = 1} \gamma_t C_t \in L(\cal{G}_j)
\end{equation}
Now choose a positive number $\epsilon$, and define 
\begin{equation*}
\beta_s := \left\{
\begin{array}{ll}
\alpha_s - \epsilon \gamma_s & \text{if } 1\le s\le m'\\
\epsilon \tilde \beta_s & \text{if } s>m'
\end{array}
\right. 
\end{equation*}
Then, for sufficiently small $\epsilon$, we have $\beta_s > 0$ for all $s=1,\ldots, m$. By the choice of $\beta_s$'s, we have
$$
\sum^{m}_{s=1} \beta_s C_s = \sum^{m'}_{s=1} (\alpha_s - \epsilon \gamma_s) C_s + \sum^m_{s = m'+1} \tilde\beta_s.
$$
After re-arranging the terms, we have 
\begin{equation*}
\sum^{m}_{s=1} \beta_s C_s = \sum^{m'}_{s=1} \alpha_s C_s + \epsilon\left (\sum^m_{s = m'+1} \tilde\beta_s C_s - \sum^{m'}_{s = 1} \gamma_s C_s \right )
\end{equation*}
By~\eqref{eq:coeff2}, we have
\begin{equation*}
A = \sum^{m'}_{s=1} \alpha_s C_s =  \sum^{m}_{s=1} \beta_s C_s \in  \Co(\cal{G}_i).
\end{equation*}

Next, we show that there exists a matrix $A\in \Co(\cal{G}_i) - \Co(\cal{G}_j)$. Note that $\Co(\cal{G}_i)$ is an open set in $L(\cal{G}_i)$ and, by assumption, $\dim L(\cal{G}_i)>\dim L(\cal{G}_j)$. Thus, there must exist a matrix $A\in \Co(\cal{G}_i) - L(\cal{G}_j)$ which implies that  $A\in\Co(\cal{G}_i) - \Co(\cal{G}_j) $.

\subsubsection{Proof for the case where $L(\cal{G}_i) \supsetneq L(\cal{G}_j)$ and $\Co(\cal{G}') \cap \Co(\cal{G}_j) = \emptyset$}
We prove this case by contradiction. Assume that there exists  a matrix $A$ such that 
\begin{equation*}
A  = \sum^{m'}_{t=1} \alpha_t C_t = \sum^m_{s =1} \beta_s C_s \in \Co(\cal{G}_i) \cap \Co(\cal{G}_j)
\end{equation*} 
for some positive coefficients $\alpha_t$'s and $\beta_s$'s. Then, 
\begin{equation*}
\sum^m_{s = m'+1} \beta_s C_s = \sum^{m'}_{s = 1} (\alpha_t - \beta_t) C_t \in \Co(\cal{G}') \cap L(\cal{G}_j)  
\end{equation*}
which is a contradiction. This completes the proof. 
\end{proof}

\end{document}